\documentclass[12pt, draftclsnofoot, onecolumn]{IEEEtran}
\usepackage{pifont}
\usepackage{}
\usepackage{mathrsfs}
\usepackage[noadjust]{cite}
\usepackage{graphicx}
\usepackage{float}
\usepackage{subfigure}
\usepackage{amsthm}
\usepackage{amsmath}
\usepackage{bm}
\usepackage{enumerate}
\usepackage{amssymb}
\usepackage{fixltx2e}
\usepackage{color}
\usepackage{hhline}
\usepackage{algpseudocode}
\usepackage{amsmath}
\usepackage{graphics}
\usepackage{epsfig}
\usepackage{setspace}
\usepackage{multirow}
\usepackage{ragged2e}
\usepackage{algorithmicx,algorithm}
\usepackage{graphics}
\usepackage{epsfig}
\newtheorem{corollary}{Corollary}
\newtheorem{lemma}{Lemma}

\allowdisplaybreaks[4]
\allowdisplaybreaks[4]
\begin{document}

\title{Next-Generation Multiple Access Based on NOMA with Power Level Modulation}
\author{\IEEEauthorblockN{Xinyue~Pei, Yingyang~Chen,
		Miaowen~Wen, 
		Hua~Yu,\\
		Erdal~Panayirci,~\IEEEmembership{Life Fellow,~IEEE,}~and~H.~Vincent~Poor,~\IEEEmembership{Life Fellow,~IEEE}}\\
		\thanks{X. Pei, M. Wen, and H. Yu are with the National Engineering Technology Research Center for Mobile Ultrasonic Detection, South China University of Technology, Guangzhou 510640, China (e-mail:
		eexypei@mail.scut.edu.cn;
		\{eemwwen, yuhua\}@scut.edu.cn).}
	\thanks{Y. Chen is with Department of Electronic Engineering, 
		College of Information Science and Technology, Jinan University, Guangzhou (e-mail:
		chenyy@jnu.edu.cn).}
	\thanks{E. Panayirci is with the Department of Electrical and Electronics
		Engineering, Kadir Has University, 34083 Istanbul, Turkey
		(e-mail: eepanay@khas.edu.tr).}
	\thanks{H. V. Poor is with the Department of Electrical and Computer Engineering, Princeton
		University, Princeton, NJ 08544 USA (e-mail: poor@princeton.edu).}
}
\markboth{SUBMITTED TO 	IEEE JOURNAL ON SELECTED AREAS IN COMMUNICATIONS}{}
\maketitle

\begin{abstract}
To cope with the explosive traffic growth of next-generation wireless communications, it is necessary to design next-generation multiple access  techniques that can provide  higher spectral efficiency as well as larger-scale connectivity. As a promising candidate, power-domain non-orthogonal multiple access (NOMA) has been widely studied. In conventional power-domain NOMA, multiple users are multiplexed in the same time and frequency band by different preset power levels, which, however, may limit the spectral efficiency under practical finite alphabet inputs. Inspired by the concept of spatial modulation, we propose to solve this problem by encoding extra information bits into the power levels, and exploit different signal constellations to help the receiver distinguish between them. To convey this idea, termed power selection (PS)-NOMA, clearly, we consider a simple downlink two-user NOMA system with finite input constellations. Assuming maximum-likelihood detection, we derive closed-form approximate bit error ratio (BER) expressions for both users. The achievable rates of both users are also derived in closed form. Simulation results verify the analysis and show that the proposed PS-NOMA outperforms conventional NOMA in terms of BER and achievable rate.
\end{abstract}

\IEEEpeerreviewmaketitle
\section{Introduction}
Future wireless networks, such as beyond fifth-generation (B5G) or sixth-generation (6G) networks, are expected to support extremely high data rates and numerous users or nodes with various applications and services [\ref{futurenetwork}]. However, the conventional orthogonal access (OMA) schemes  used in the previous wireless generations cannot meet these unprecedented demands, limiting the improvement of the overall spectral efficiency (SE). Against the background, non-orthogonal multiple access (NOMA) was proposed, which allows numerous users to share the same resource (e.g., a time/frequency resource block) and separate the users in  power and code domains at the expense of  additional receiver complexity [\ref{NOMAconcept}], [\ref{YWLIU1}]. In the power-domain NOMA, users are typically multiplexed with different  power levels by using superposition coding at the transmitter and are distinguished through successive interference cancellation (SIC) at the receiver [\ref{SIC}], [\ref{HVPOOR}]. Compared with OMA, NOMA has higher significant system throughput  and greater fairness [\ref{NOMAad1}]-[\ref{NOMAad3}]. Due to these advantages,  NOMA has been recognized as a critical technology by the third-generation partnership project (3GPP) for future wireless networks [\ref{3GPP}].

It is worth noting that the majority of existing 
NOMA schemes assume Gaussian input signals [\ref{Gaussian1}]-[\ref{Gaussian3}]. Although Gaussian inputs can theoretically attain the channel capacity, the corresponding
implementation faces many difficulties, e.g., very large storage capacity, high  computational complexity, and extremely long
decoding delay [\ref{disad_gauss}]. Applying the results
derived from the Gaussian inputs to the signals with finite alphabet
inputs, e.g., pulse amplitude
modulation (PAM), can result in a significant performance loss [\ref{loss_gauss}]. Motivated by this situation, some researchers have considered finite alphabet inputs [\ref{two_user_GMAC}]-[\ref{uplink_NOMA_finite}]. In particular, the authors of [\ref{two_user_GMAC}] studied NOMA-aided two-user Gaussian multiple
access channels (MACs) with finite complex input constellations, and the constellation-constrained capacity region of the proposed scheme was derived. Moreover, they designed a constellation rotation (CR) scheme for $M$-ary phase shift keying (PSK) and $M$-PAM signals. In [\ref{two_GMAC_PA}], the same authors proposed a novel power allocation scheme for the  model in [\ref{two_user_GMAC}], which can achieve similar performance as the CR scheme while reducing the decoding complexity for quadrature amplitude modulation (QAM) constellations. The authors of [\ref{up_noma_finite}] developed a novel 
framework for a classical two-transmitter two-receiver NOMA system over Z-channels with QAM and
max-min user fairness based on the aforementioned two studies.  Specifically, they formulated a max-min
optimization problem to maximize the smaller minimum Euclidean
distance among the two resulting signal constellations at both
receivers. Similarly, in  [\ref{uplink_NOMA_finite}], the authors considered a classical two-user MAC with NOMA and practical QAM constellations, aiming at 
maximizing the minimum Euclidean distance of the received sum
constellation with a maximum likelihood (ML) detector by controlling transmitted powers
and phases of users.

On the other hand, spatial modulation (SM) has
been regarded as another promising multi-antenna technique of improving
SE for next-generation wireless networks [\ref{SM}]. Unlike conventional multiple-input multiple-output (MIMO), SM selects only one activated antenna for each transmission, thus avoiding inter-antenna
interference and the requirement of multiple radio frequency (RF) chains. Hence, SM systems can  achieve reduced
implementation cost and complexity. Precisely, in SM systems, the transmitted information consists of the index of an
active antenna and a  modulated symbol. Clearly, combining NOMA with SM (termed SM-NOMA) will further improve system SE without increasing the power consumption and implementation
complexity. Therefore, SM-NOMA has been exploited extensively in recent years, and capacity or BER performance analysis is the focus of many works [\ref{rate_sm_noma_opt}]-[\ref{SM_coNOMA}]. In [\ref{rate_sm_noma_opt}], an iterative algorithm for the
spatial-domain design was proposed to
maximize the instantaneous capacity of SM-NOMA.  The authors of [\ref{NOMA_SM_V2V_MIMO}] applied SM-NOMA in wireless vehicle-to-vehicle (V2V) environments. Moreover, they analyzed  closed-form capacity expressions, and  BER performance via Monte Carlo simulations and formulated a pair of power allocation
optimization schemes for the system.
 In [\ref{SM_coNOMA}], the authors studied the BER and capacity performance of a novel three-node cooperative relaying system using SM-aided NOMA. Notably, all the above mentioned works focused on Gaussian input, which do not apply to  finite input constellations.  Few researchers have considered finite alphabet inputs in SM-NOMA systems [\ref{finite_SM_NOMA_letter}]-[\ref{fund_SM_finite_NOMA}]. The authors of [\ref{finite_SM_NOMA_letter}], [\ref{finite_SM_NOMA}] proposed and studied the respective SM-NOMA systems from their mutual information (MI) perspective. Since the MI lacks a closed-form formulation, they proposed a lower bound to quantify it.
  In [\ref{SM_NOMA_finite_ICC}], an SM aided cooperative NOMA scheme with bit allocation  was studied, and its SE was analyzed.
  The authors of [\ref{fund_SM_finite_NOMA}] mainly investigated the fundamental applicability of SM for multi-antenna channels and found intrinsic cooperation for constructing
 energy-efficient finite-alphabet NOMA.

However, to the best of our knowledge, all works about SM-NOMA  simply considered transmitting NOMA signals 
in SM systems. In other words, they did not change NOMA itself. Conversely, how can we use an SM-like idea to design a novel NOMA scheme and improve its performance? Furthermore, will we be able to design the transmit power to achieve effects similar to the spatial gain in SM-NOMA? At the time of writing, no work has solved this problem.  Against the background,  in this paper,
we propose a novel downlink NOMA scheme using power selection (PS) with finite-alphabet inputs, termed  PS-NOMA, over Rayleigh fading channels, which consists of one BS and two users. At the transmitter, except for the PAM symbols intended for the users, we also design the transmit power into the codebook, thus increasing the number of bits transmitted.  The main contributions
of this paper are summarized as follows:
\begin{itemize}
	\item We design a novel PS-NOMA scheme for the classical
	two-user downlink channel, where the transmitter randomly chooses transmit power from the preset power matrix. Nevertheless, PS makes both transmission and decoding different from previous works. To this end, we design specific constellations and the ML detector for the proposed scheme. We further discuss the influence caused by different configurations of the power matrix and give the optimal order of the matrix. Moreover, we calculate the best design for the power matrix based on the minimum Euclidean distance of the constellation. 
	\item The performance of PS-NOMA is theoretically analyzed
	in terms of BER and achievable rate. We derive closed-form achievable rate expressions for PS-NOMA. Since it is difficult to derive exact BER expressions, approximate closed-form BER
	expressions are derived for users employing $M$-ary PAM instead, which well match with the simulation counterparts. It is worth noting that all derived results are restricted
	to specific constellations and Rayleigh fading channels.
	\item Monte Carlo simulations are performed to 
	verify the theoretical analysis. In consideration of  fairness, conventional  NOMA without PS (termed NOMA) is used as a benchmark.  Simulation
	results  show
	that PS-NOMA outperforms NOMA
	in terms of BER and achievable rate.
\end{itemize}

The rest of this paper is organized as follows. Section~II describes the system model of PS-NOMA.  Section \uppercase\expandafter{\romannumeral3} provides the achievable rate and BER analyses.  Section \uppercase\expandafter{\romannumeral4} analyzes the numerical results. Finally, we conclude the work in Section~\uppercase\expandafter{\romannumeral5}. 

\emph{Notation}: The probability of
an event and the probability density function (PDF) of a random variable are denoted
by $\Pr(\cdot)$ and $p(\cdot)$, respectively. $Q(\cdot)$,
$\Gamma(\cdot)$, $\mathbb{E}\{\cdot\}$, and $\text{Var}\{\cdot\}$ denote the Gaussian Q-function, gamma function,
expectation, and variance, respectively. $|\cdot|$ denotes the absolute value of
a complex scalar. $\lfloor \cdot \rfloor$ denotes the floor function. $\mathcal{R}\{\cdot\}$ indicates the real part of a complex value. Superscript $*$ stands for complex conjugates. ${\rm{I}}(\cdot)$ denotes MI, and ${\rm H}(\cdot|\cdot)$ denotes entropy. Finally, $x\sim\mathcal N_c(\mu_x,\beta_x)$ indicates that the random
variable $x$ obeys a complex Gaussian distribution with mean
$\mu_x$ and variance $\beta_x$. 
\section{System Model}
\begin{figure*}[t]
	\centering
	\includegraphics[width=7in]{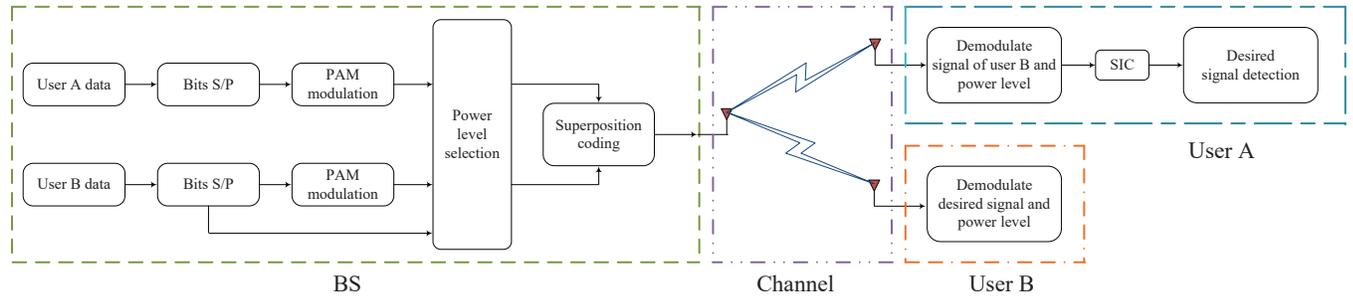}
	\caption{System model with 2 users and $N=2$.} 	
	\label{system_model}
\end{figure*}
\begin{figure}[t]
	\centering
		\includegraphics[width=5in]{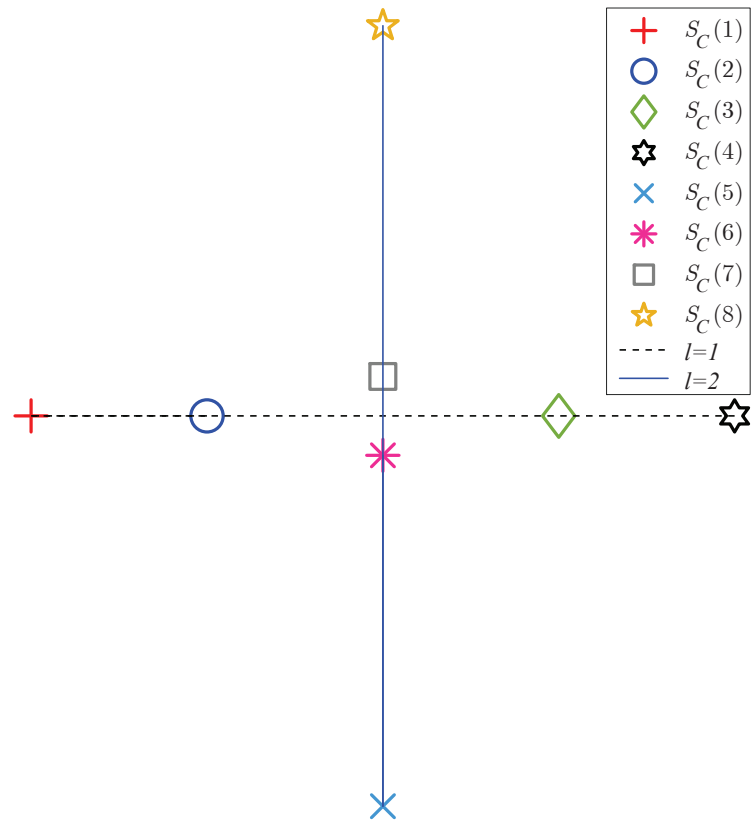}
	\caption{Constellation $\mathcal{S}_C(i)$ ($i\in\{1,\cdots,8\}$) of superimposed signal transmitted with 2-PAM signals $s_A$ as well as  $s_B$ and $N=2$.}
	\label{constell_superimposed_trans}
\end{figure} 

\subsection{Transmission Model}
Before introducing PS-NOMA, let us recall the conventional downlink two-user NOMA system first. In such a system, the base station (denoted by BS) transmits a  superimposed signal to the users. Based on the NOMA protocol, this signal consists of the data symbols desired by both users, and the power allocated to the symbols. In conventional NOMA systems, information is carried out only by data symbols. However, the power allocation coefficients contained in the superimposed signal can also be utilized to carry information by appropriate design.
  
Based on this idea, we consider a downlink PS-NOMA network as illustrated in Fig. \ref{system_model}, which consists of a BS, a  near user (denoted by $U_A$), a far user (denoted
by $U_B$). All the nodes are equipped with a single antenna. The channels of BS$\rightarrow U_A$ and BS$\rightarrow U_B$
are respectively denoted by $h_A$ and $h_B$, which are assumed to follow $\mathcal N_c(0,\beta_A)$ and $\mathcal N_c(0,\beta_B)$ distributions. Notably, in our system, all the nodes are
assumed to be synchronous in both time and
frequency. Moreover, complete and perfect channel state information (CSI) is assumed to be
available to users.  Unlike conventional NOMA, we obtain a spectral efficiency  gain by selecting a power level (PL) at the transmitter. Similar to the SM, a
PL is chosen from a set of PLs $N$ at random where $N=2^k$, $k=1,2,\ldots$, so that each PL
carries $\log_2(N)$ bits of information in addition to the amount of information carried by the transmitted data symbols.  Noticing that $U_A$ is the near user and $U_B$ is the far
user, we have $\beta_A\geq\beta_B$, and the power allocated to $U_A$ should be smaller than that assigned to $U_B$ according to the principle of NOMA [\ref{SIC}]. In this manner, the transmit power allocation coefficient matrix can be written as
\begin{align}      
	\textbf{P}=\left[                 
	\begin{array}{cc}   
		\sqrt{p_A(1)}, & \sqrt{1-p_A(1)}\\  
		\sqrt{p_A(2)}, & \sqrt{1-p_A(2)}\\  
		\cdots &\cdots\\
		\sqrt{p_A(N)}, & \sqrt{1-p_A(N)}\\
	\end{array}
	\right],                 
\end{align}
where the first  ($0<p_A(l)<0.5$, $l\in\{1,\ldots,N\}$) and the second columns represent the power allocation coefficient for $U_A$ and  $U_B$, respectively, and different rows represent different PLs. To ensure the PLs are distinguishable,   we deliberately rotate each PL with a certain angle, which is given by $\bm{\Theta}=\text{diag}(\exp(0),\exp(j\pi/N),\cdots,\exp(j(N-1)\pi/N))$. Therefore, the rotated power allocation coefficient matrix can be derived as
\begin{align}
	\textbf{G}=\bm{\Theta}\textbf{P}=\left[                 
	\begin{array}{cc}   
		\alpha_A(1), & \alpha_B(1)\\  
		\alpha_A(2), & \alpha_B(2)\\
		\cdots&\cdots\\
		\alpha_A(N), & \alpha_B(N)\\
	\end{array}
	\right],
\end{align}
where $\alpha_A(l)=\exp(j(l-1)\pi/N)p_A(l)$ and $\alpha_B(l)=\exp(j(l-1)\pi/N)(1-p_A(l))$.

During each transmission, the BS first chooses the PL $l\in \mathcal{N}$ $(\mathcal{N}={1,2,\cdots,N})$ as the transmit power. Subsequently, it conveys the superposition
coded symbol 
\begin{align}
	s_C=\alpha_A(l)s_A+\alpha_B(l)s_B
\end{align}
to $U_A$ and $U_B$ simultaneously, where $s_A\in \mathcal{S}_A$ and $s_B\in \mathcal{S}_B$ are the data symbols intended for
$U_A$ and $U_B$, respectively. We assume that $s_A$ and $s_B$ are $M_A$-ary PAM and $M_B$-ary PAM symbols with $\mathbb{E}\{|s_A|^2\}=\mathbb{E}\{|s_B|^2\}=1$;  
\begin{align}
\mathcal{S}_A=\{\pm d_A, \pm 3d_A,\cdots, \pm (2M_A-1)d_A \}
\end{align}
and
\begin{align}
\mathcal{S}_B=\{\pm d_B, \pm 3d_B,\cdots, \pm (2M_B-1)d_B \}
\end{align} are the corresponding $M_A$ and $M_B$ points constellations, respectively, where $d_A=\sqrt{3/(M_A^2-1)}$ and $d_B=\sqrt{3/(M_B^2-1)}$ respectively represent  half of the minimum distance between two adjacent points of the  normalized  $M_A$-PAM and $M_B$-PAM constellations.  Let $\mathcal{M}_i=\{1,2,\cdots,M_i\}$ ($i\in\{A,B\}$), the signal constellation of
$s_C$ can be expressed as 
\begin{align}
\mathcal{S}_C=\{s_C|s_i=\mathcal{S}_i(k_i), {\rm{PL}}=l, k_i\in \mathcal{M}_i, l\in \mathcal{N}, i\in\{A,B\}\},
\end{align}
which is an irregular $M_AM_BN$-ary constellation, as shown in Fig. \ref{constell_superimposed_trans}.  Consequently, the received signals at $U_A$ and $U_B$ can be given by
\begin{align}
y_A=h_As_C+n_A 
\end{align}
and 
\begin{align}
	y_B=h_Bs_C+n_B,
\end{align}
respectively, where $n_A\sim\mathcal{N}_c(0,N_0)$ ($n_B\sim\mathcal{N}_c(0,N_0)$) indicates the  additive white Gaussian noise (AWGN) at $U_A$ ($U_B$) with power spectral density $N_0$.
\begin{figure}[t]
	\centering
	\includegraphics[width=5in]{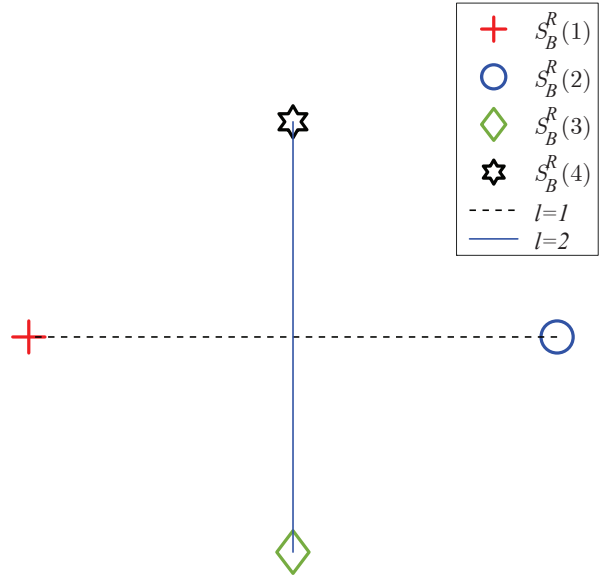}
	\caption{Constellations of $s_B$ with $N=2$, where  $s_B$ is a 2-PAM symbol,  and $\mathcal{S}_B^R$ denotes the union constellation of $s_B$ and PL. $\mathcal{S}_B^R(1),\mathcal{S}_B^R(3) : s_B=-1$, $\mathcal{S}_B^R(2), \mathcal{S}_B^R(4): s_B=1$.}
	\label{constellation_symB}
\end{figure}  
\begin{figure}[t]
	\centering
	\includegraphics[width=5in]{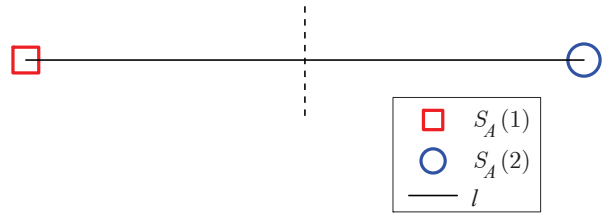}
	\caption{Constellation of 2-PAM signal $s_A$.}
	\label{constell_symA}
\end{figure} 
\begin{algorithm}[t]
	\caption{Detection Algorithm} 
	\label{Detection_algorithm} 
	\begin{algorithmic}[1]
		\For {$U_A$}   	
		\State Compare the received signal $y_A$ with the
		joint constellation  $\mathcal{S}_B^R$:
		\begin{align}\label{ml_UA_sB}
			\mathcal{S}_B^{R}(\hat{i})=\mathop {\arg\min\limits_{i}}|y_A-h_A\mathcal{S}_B^R(i)|^2, 
		\end{align}
		where  $i\in\{1,\cdots,NM_B\}$, and $\hat{i}$ is the estimate of $i$ at $U_A$.
		\State Determine active PL $\hat{l}$ and the transmitted signal $\hat{s}_B$ by resolving $\mathcal{S}_B^{R}(\hat{i})$, where $\hat{s}_B$ and $\hat{l}$ are the estimates of $s_B$ and $l$ at $U_A$. 
		\State Use SIC to decode $s_A$,	whose ML detector can be expressed as
		\begin{align}\label{ML_UA_SA}
			\hat{s}_A=\mathop {\arg\min\limits_{s_A}}|y_A-h_A\alpha_B(\hat{l})\hat{s}_B-h_A\alpha_A(\hat{l})s_A|^2,
		\end{align}
		where $\hat{s}_A$ is the estimate  of $s_A$ at $U_A$.
		\EndFor
		\For {$U_B$}
		\State Compare the received signal $y_B$ with the
		joint constellation $\mathcal{S}_B^R$:
		\begin{align}
			\mathcal{S}_B^{R}(\hat{i})=\mathop {\arg\min\limits_{i}}|y_B-h_A\mathcal{S}_B^R(i)|^2, 
		\end{align}
		where  $i\in\{1,\cdots,NM_B\}$, and $\hat{i}$ is the estimate of $i$ at $U_B$.
		\State Determine active PL $\hat{l}$ and the transmitted signal $\hat{s}_B$ by resolving $\mathcal{S}_B^{R}(\hat{i})$, where $\hat{s}_B$ and $\hat{l}$ are the estimates of $s_B$ and $l$ at $U_B$. 
		\EndFor  
	\end{algorithmic} 
\end{algorithm}
\subsection{Signal Decoding}
The above subsection describes the  transmission model of PS-NOMA. Now we shift our focus to the decoding process.
Before the following description, we need to first clarify which data symbol gets the information of the PL. \emph{In this paper, we use the information of PL to enhance the performance of $U_B$. In other scenarios, it can also be used to improve the performance of $U_A$}.

It is assumed that both users invoke ML detectors to decode their signals. As we know, PL is an important factor affecting the decoding procedure because it is closely related to SIC and the decoding of data symbols. In conventional NOMA scheme, whose PLs are fixed, the users already know the values of PLs when decoding. While in PS-NOMA, PLs change over time and are unknown at the users.  Hence, the conventional ML detector cannot be applied directly in PS-NOMA, and we should also decode PL at both users. Under the NOMA protocol, $U_A$ first decodes $s_B$ and $l$; for $U_B$, it detects $s_B$ and $l$ by treating $s_A$ as noise. We herein use the joint constellation of $s_B$ and PL to help decoding. For example, in the case
of $N=2$ and 2-PAM $s_B$ as shown in Fig. \ref{constellation_symB}, a combined 4-ary constellation is formed with
four points, i.e.,
\begin{align}
\mathcal{S}_B^R(i)=\{s_B=\pm 1, l=1,\,2\}, i\in\{1,\cdots,4\}. 
\end{align} 
Explicitly, the information of $s_B$ and PL is contained in each point of the constellation. If the received signal after equalization is close  to $\mathcal{S}_B^R(1)$, we can readily know the estimated $s_B=-1$ and $l=1$. With these information, $U_A$ then subtracts the estimated $s_B$ and decodes its own signal, and this procedure can be seen as a common 2-PAM signal decoding as shown in Fig. \ref{constell_symA}. The detailed detection algorithm for both users are summarized in  \textbf{Algorithm \ref{Detection_algorithm}}. 

\section{Performance Analysis}
In this section, we theoretically analyze the performance of
the proposed PS-NOMA scheme in terms of 
achievable rate and BER. Moreover, we discuss the optimal design of $N$ and $\textbf{G}$.

\subsection{Achievable Rate Analysis}

For convenience, we let $\{a,b,l\}$ ($a\in \mathcal{M}_A$, $b\in \mathcal{M}_B$, $l\in \mathcal{N}$) denote  $(\alpha_A(l)\mathcal{S}_A(a)+\alpha_B(l)\mathcal{S}_B(b))$. 
It is considered that all PLs are activated equally likely for PS-NOMA.
\subsubsection{Achievable Rate Analysis of $U_B$} It follows from the chain rule that
\begin{align}\label{ISBYB}	
{\rm{I}}(s_B,{\rm{PL}};y_B)={\rm{I}}(s_B;y_B|{\rm{PL}})\it+{\rm{I}}({\rm{PL}}\it;y_B).
\end{align}
 Therefore, we deal with the calculations of ${\rm{I}}(s_B;y_B|{\rm{PL}})$ and ${\rm{I}}({\rm{PL}};y_B)$ in the following. Specifically, we have 
\begin{align}\label{ISBYB_PL}
{\rm{I}}(s_B;y_B|{\rm{PL}})= {\rm{H}}(\it y_B|{\rm{PL}})-{\rm{H}}(\it y_B|s_B,{\rm{PL}}),
\end{align}
where
\begin{align}
{\rm{H}}(y_B|{\rm{PL}})=\frac{1}{N}\sum_{l=1}^{N}{\rm{H}}(y_B|{\rm{PL}}=l)
\end{align}
and
\begin{align}	
{\rm{H}}(y_B|s_B,{\rm{PL}})=\frac{1}{NM_B}\sum_{k_B=1}^{M_B}\sum_{l=1}^{N}{\rm{H}}(y_B|s_B=\mathcal{S}_B(k_B),{\rm{PL}}=l).
\end{align}
Here, ${\rm{H}}(y_B|{\rm{PL}}=l)$ and ${\rm{H}}(y_B|s_B=\mathcal{S}_B(k_B),{\rm{PL}}=l)$ can be respectively given by	
\begin{align}\label{HyBPl}
	&{\rm{H}}(y_B|{\rm{PL}}=l)=\int p(y_B|{\rm{PL}}=l)\log_2(p(y_B|{\rm{PL}}=l))dy_B\nonumber\\
	&=-\frac{\sum_{k_A=1}^{M_A}\!\sum_{k_B=1}^{M_B}\!\sum_{l=1}^{N}\!\log_2\left({\sum_{i_A=1}^{M_A}\!\sum_{i_B=1}^{M_B}\!p(y_B|s_A=\mathcal{S}_A(i_A),s_B=\mathcal{S}_B(i_B),{\rm{PL}}=l)}/{M_AM_B} \right)}{M_AM_BN}
\end{align}	
and
\begin{align}\label{HyBSBPl}
	&H(y_B|s_B\!=\!\mathcal{S}_B(k_B),{\rm{PL}}\!=\!l)\!\!=\!\!\int\!\! p(y_B|s_B\!=\!\mathcal{S}_B(k_B),\!{\rm{PL}}\!=\!l)\log_2(p(y_B|s_B=\mathcal{S}_B(k_B),{\rm{PL}}=l))dy_B\nonumber\\
	&=-\frac{\sum_{k_A=1}^{M_A}\sum_{k_B=1}^{M_B}\sum_{l=1}^{N}\log_2\left({\sum_{i_A=1}^{M_A}p(y_B|s_A=\mathcal{S}_A(i_A),s_B=\mathcal{S}_B(k_B),{\rm{PL}}=l)}/{M_A}\right)}{M_AM_BN}.
\end{align}	 Substituting (\ref{HyBPl}), (\ref{HyBSBPl}), and 
\begin{align}\label{py_iaibl}
	p(y_B|s_A=\mathcal{S}_A(i_A),s_B=\mathcal{S}_B(i_B),{\rm{PL}}=l)=\frac{1}{\pi N_0}\exp\left(-\frac{|y_B-h_B\{i_A,i_B,l\}|^2}{N_0}\right)
\end{align}
 into (\ref{ISBYB_PL}), ${\rm{I}}(s_B;y_B|{\rm{PL}})$ can be given by
\begin{align}\label{eq_of_IsBYBPL}
	&{\rm{I}}(s_B;y_B|{\rm{PL}})=\log_2(M_B)-\frac{1}{M_AM_BN}\nonumber\\
	&\times\sum_{k_A=1}^{M_A}\sum_{k_B=1}^{M_B}\sum_{l=1}^{N}\mathbb{E}\left\{\log_2\left(\frac{\sum_{i_A=1}^{M_A}\sum_{i_B=1}^{M_B}\exp(-|h_B\{k_A,k_B,l\}-h_B\{i_A,i_B,l\}+n_B|^2/N_0)}{\sum_{i_A=1}^{M_A}\exp(-|h_B\{k_A,k_B,l\}-h_B\{i_A,k_B,l\}+n_B|^2/N_0)}\right)\right\}
	.
\end{align}
 In this manner, ${\rm{I}}({\rm{PL}};y_B)$ can also be derived as 	\begin{align}\label{IPL_YB}
 	&{\rm{I}}({\rm{PL}};y_B)={\rm{H}}(y_B)-{\rm{H}}(y_B|{\rm{PL}})\nonumber\\
 	&=\log_2(N)+\frac{1}{M_AM_BN}\nonumber\\
 	&\times\sum_{k_A=1}^{M_A}\!\sum_{k_B=1}^{M_B}\!\sum_{l=1}^{N}\mathbb{E}\left\{\log_2\left(\frac{\sum_{i_A=1}^{M_A}\!\sum_{i_B=1}^{M_B}\!\exp(\!-\!|h_B\{k_A,k_B,l\}\!\!-\!\!h_B\{i_A,i_B,l\}\!\!+\!\!n_B|^2/N_0)}{\sum_{i_A=1}^{M_A}\!\sum_{i_B=1}^{M_B}\!\sum_{ll=1}^{N}\!\exp(\!-\!|h_B\{k_A,k_B,l\}\!\!-\!\!h_B\{i_A,i_B,ll\}\!\!+\!\!n_B|^2/N_0)}\right)\right\}.
 \end{align}
\subsubsection{Achievable Rate Analysis of $U_A$}
Different from $U_B$, $U_A$ detects
its own signal after removing the interference imposed by $U_B$. By invoking a similar mathematical  method, the achievable rate for
$U_A$ detecting $s_B$ is given by
\begin{align}\label{ISBYA}
	{\rm{I}}(s_B,{\rm{PL}};y_A)={\rm{I}}(s_B;y_A|{\rm{PL}})+{\rm{I}}({\rm{PL}};y_A),
\end{align}
where ${\rm{I}}(s_B;y_A|{\rm{PL}})$ and ${\rm{I}}({\rm{PL}};y_A)$ can be respectively derived as
\begin{align}\label{eq_of_IsBYAPL}
	&{\rm{I}}(s_B;y_A|{\rm{PL}})=\log_2(M_B)-\frac{1}{M_AM_BN}\nonumber\\
	&\times\sum_{k_A=1}^{M_A}\sum_{k_B=1}^{M_B}\sum_{l=1}^{N}\mathbb{E}\left\{\log_2\left(\frac{\sum_{i_A=1}^{M_A}\sum_{i_B=1}^{M_B}\exp(-|h_A\{k_A,k_B,l\}-h_A\{i_A,i_B,l\}+n_A|^2/N_0)}{\sum_{i_A=1}^{M_A}\exp(-|h_A\{k_A,k_B,l\}-h_A\{i_A,k_B,l\}+n_A|^2/N_0)}\right)\right\}
\end{align}
and 	
\begin{align}\label{IPL_YA}
	&{\rm{I}}({\rm{PL}};y_A)=\log_2(N)+\frac{1}{M_AM_BN}\nonumber\\
	&\times\!\!\sum_{k_A=1}^{M_A}\!\sum_{k_B=1}^{M_B}\!\sum_{l=1}^{N}\!\mathbb{E}\left\{\log_2\left(\frac{\sum_{i_A=1}^{M_A}\!\sum_{i_B=1}^{M_B}\!\exp(\!-\!|h_A\{k_A,k_B,l\}\!-\!h_A\{i_A,i_B,l\}\!+\!n_A|^2\!/\!N_0)}{\sum_{i_A=1}^{M_A}\!\sum_{i_B=1}^{M_B}\!\sum_{ll=1}^{N}\!\exp(-|h_A\{k_A,k_B,l\}\!-\!h_A\{i_A,i_B,ll\}\!+\!n_A|^2\!/\!N_0)}\right)\right\}.
\end{align} With the detected PL, the achievable rate for
$U_A$ detecting $s_A$, i.e., ${\rm{I}}(s_A;y_A|{\rm{PL}})$ can be derived as 	\begin{align}\label{ISA_YA_PL}
	&{\rm{I}}(s_A;y_A|{\rm{PL}})=\log_2(M_A)-\frac{1}{M_AN}\nonumber\\
	&\times{\sum_{k_A=1}^{M_A}\sum_{l=1}^{N}}\mathbb{E}\left\{\log_2\left(\frac{\sum_{i_A=1}^{M_A}\exp\left(|h_A\alpha_A(l)(\mathcal{S}_A(k_A)-\mathcal{S}_A(i_A))+n_A|^2/N_0\right)}{\exp\left(-|n_A|^2/N_0\right)}\right)\right\}.
\end{align}

\subsection{BER Analysis}
In this subsection, closed-form BER expressions are derived for $U_A$ and $U_B$ in PS-NOMA systems.

\subsubsection{BER of $U_B$}
To obtain the BER of $U_B$, we should  derive the symbol error ratio (SER) of  $s_B$ at $U_B$  first. Similar to [\ref{PEP_NOMA}], [\ref{pepmethod}], we utilize the conditional pairwise error probability (PEP), which indicates the probability of detecting $s_B$ as $\hat{s}_B$ conditioned
on $h_B$, namely 
\begin{align}\label{PEP_SB_HB}
	&\Pr(s_B\to\hat{s}_B|h_B)\nonumber\\
	&=\Pr(|y_B-h_B\alpha_B(l)s_B|^2>|y_B-h_B\alpha_B(\hat{l})\hat{s}_B|^2)\nonumber\\
	&\overset{(a)}{=}\Pr(|h_Bx_B|^2-|h_B\hat{x}_B|^2-2\mathcal{R}\{y_B^*h_B(x_B-\hat{x}_B)\}>0)\nonumber\\
	&\overset{(b)}{=}\Pr(-|h_B\triangle x_B|^2-2\mathcal{R}\{w_B^*h_B\triangle x_B\}>0)\nonumber\\
	&\overset{(c)}{=}Q\left(\sqrt{\frac{|h_B|^2(|\triangle x_B|^2+2\mathcal{R}\{x_A^*\triangle x_B\})^2}{2N_0|\triangle x_B|^2}}\right)\nonumber\\
	&\overset{(d)}{=}Q\left(\sqrt{\frac{|h_B|^2}{2N_0^{\rm {equal}}}}\right),
\end{align}
where we define in $(a)$  ${x}_B\triangleq\alpha_B({l}){s}_B$ and $\hat{x}_B\triangleq\alpha_B(\hat{l})\hat{s}_B$; in $(b)$  $w_B=h_B\alpha_A(l)s_A+n_B=h_Bx_A+n_B$, with $x_A\triangleq\alpha_A({l}){s}_A$, and $\triangle x_B=x_B-\hat{x}_B$; the detailed derivation of $(c)$  is present in Appendix \ref{Proof_of_c}; $(d)$ is obtained by defining equivalent noise variance as
\begin{align}\label{equival_noise}
	N_0^{\rm {equal}}\triangleq\frac{N_0|\triangle x_B|^2}{(|\triangle x_B|^2+2\mathcal{R}\{x_A^*\triangle x_B\})^2}.
\end{align} Since we model $h_B$ as a complex Gaussian random variable, $|h_B|^2$ has a PDF
\begin{align}
	p_{|h_B|^2}(x)=\frac{1}{\beta_B}\exp\left(-\frac{x}{\beta_B}\right).
\end{align}
With $p_{|h_B|^2}(x)$, we can arrive  at [\ref{Psbtohatsb}, Eq. (64)]
\begin{align}\label{PEP_of_xB_uB}
		\setcounter{equation}{28}
	\Pr(s_B\to\hat{s}_B)&=\int_0^{+\infty}\Pr(s_B\to\hat{s}_B|h_B)p_{|h_B|^2}(x)dx\nonumber\\
	&= \int_0^{+\infty}Q\left(\sqrt{\frac{|h_B|^2}{2N_0^{\rm {equal}}}}\right)p_{|h_B|^2}(x)dx\nonumber\\
	&=\frac{1}{2}\left(1-\sqrt{\frac{\beta_B}{4N_0^{\rm {\rm {equal}}}+\beta_B}}\right),
\end{align}
which can be averaged over all the possible values of $x_A^*$ to consider all interference scenarios. Notably, $x_A$ herein should keep the same PL as $x_B$.
Let $P_{ij,B}^{\rm {bit}}$ denote the number of error bits when symbol $\mathcal{S}_B^R(i)$ is incorrectly detected as $\mathcal{S}_B^R(j)$. Since $s_B$ and PL have $\log_2(NM_B)$ bits, the BER of $s_B$ and PL detections at $U_B$ can be approximated by
\begin{align}\label{PB}
	P_B&\approx\frac{1}{NM_B\log_2(NM_B)}\sum_{i=1}^{NM_B}\sum_{j \neq i}\Pr(\mathcal{S}_B^R(i)\to\mathcal{S}_B^R(j))P_{ij,B}^{\rm {bit}}.
\end{align}

\subsubsection{BER of $U_A$}
Recall the ML detection in (\ref{ml_UA_sB}) and (\ref{ML_UA_SA}). $U_A$ first decodes PL as well as $s_B$ and extracts $s_B$ from the received signal
$y_A$ and then decodes $s_A$ from the residual signal. Obviously,  
the overall BER can be divided into two complementary
parts, depending on whether $s_B$ and PL $l\in\mathcal{N}$ are correctly detected or not. Since we have $N$ PLs, the overall BER of $U_A$ can
be evaluated as
\begin{align}\label{PA}
	P_A\approx(1-P_A^{s_B})\left(\frac{1}{N}\sum_{l=1}^{N}P_{\rm {bit}}^A(l)\right)+\frac{P_A^{s_B}}{2},
\end{align}
where $P_{A}^{s_B}$
is the SER of $s_B$ at $U_A$, and $P_{\rm {bit}}^A(l)$ denotes the BER of $M_A$-ary PAM  demodulation
over Rayleigh fading channels with PL $l$ selected.  Since it is difficult
to exactly calculate the BER when $s_B$ or PL is detected
incorrectly at $U_A$, we use the estimate ${P_A^{s_B}}/{2}$ instead. Similarly, to derive $P_{A}^{s_B}$, we use the conditional PEP, which denotes the probability of detecting $s_B$ as $\hat{s}_B$ conditioned
on $h_A$ [\ref{PEP_NOMA}], [\ref{pepmethod}] 
\begin{align}
	&\Pr(s_B\to\hat{s}_B|h_A)\nonumber\\
	&=\Pr(|y_A-h_A\alpha_B(l)s_B|^2>|y_A-h_A\alpha_B(\hat{l})\hat{s}_B|^2)\nonumber\\
	&=Q\left(\sqrt{\frac{|h_A|^2(|\triangle x_B|^2+2\mathcal{R}\{x_A^*\triangle x_B\})^2}{2N_0|\triangle x_B|^2}}\right)\nonumber\\
	&=Q\left(\sqrt{\frac{|h_A|^2}{2N_0^{\rm {\rm {equal}}}}}\right).
\end{align}
Definitions of ${x}_B$, $\hat{x}_B$, $x_A$, $\triangle x_B$, and $N_0^{\rm {equal}}$ are given in (\ref{PEP_SB_HB}) and (\ref{equival_noise}). Having
\begin{align}
	p_{|h_A|^2}(x)=\frac{1}{\beta_A}\exp\left(-\frac{x}{\beta_A}\right),
\end{align}
we can arrive at [\ref{Psbtohatsb}, Eq. (64)]
\begin{align}\label{PEP_of_xB_ua}
	\Pr(s_B\to\hat{s}_B)&=\int_0^{+\infty}\Pr(s_B\to\hat{s}_B|h_A)p_{|h_A|^2}(x)dx\nonumber\\
	&=\frac{1}{2}\left(1-\sqrt{\frac{\beta_A}{4N_0+\beta_A}}\right),
\end{align}
which can be averaged over all the possible values of $x_A^*$ to consider all interference scenarios. Equally, the $x_A$ herein should keep the same PL as $x_B$. Hence, $P_A^{s_B}$ is given by
\begin{align}\label{PA_SB}
	P_A^{s_B}\approx \frac{1}{NM_B}\sum_{i=1}^{NM_B}\sum_{j\neq i}\Pr(\mathcal{S}_B^R(i)\to\mathcal{S}_B^R(j)).
\end{align}
As for the calculation of $P_{\rm {bit}}^A(l)$, we first derive the conditional error probability for PL $l$ under the $m_A$-th bit ($m_A\in\{1,\ldots,\log_2(M_A)\}$)  and the conditions of receive signal-to-noise ratio (SNR) at $U_A$, i.e.,
\begin{align}
	\gamma_A=\frac{\alpha_A(l)^2|h_A|^2}{N_0},
\end{align}
as [\ref{PAB_paper}] 
\begin{align}\label{PAbit_m}
	&P_{\rm {bit}}^A(l|m_A,\gamma_A)=\frac{2}{M_A}\nonumber\\
	&\times\sum_{i=0}^{(1-2^{-m_A})M_A-1}\left\{(-1)^{\lfloor {i(2^{m_A}-1)}/{M_A}\rfloor}\left(2^{m_A-1}-\lfloor\frac{i(2^{m_A-1})}{M}+\frac{1}{2}\rfloor\right)Q\left((2i+1)d_A\sqrt{2\gamma_A}\right)\right\}.
\end{align}
 By substituting 
the alternative representation, i.e.,
\begin{align}
	Q(x)=\frac{1}{\pi}\int_{0}^{\frac{\pi}{2}}\exp\left(-\frac{x^2}{2\sin^2\psi}\right)d\psi
\end{align}
for the Gaussian Q-function in (\ref{PAbit_m}), we can obtain
\begin{align}
P_{\rm {bit}}^A(l|m_A,\gamma_A)=&\frac{2}{M_A}\sum_{i=0}^{(1-2^{-m_A})M_A-1}\{(-1)^{\lfloor \frac{i(2^{m_A}-1)}{M_A}\rfloor}\nonumber\\
	&\times(2^{m_A-1}-\lfloor\frac{i(2^{m_A-1})}{M}+\frac{1}{2}\rfloor)\nonumber\\
	&	\times \underbrace{\frac{1}{\pi}\int_{0}^{\pi/2}\exp\left(-\frac{((2i+1)d_A)^2\gamma_A}{\sin^2\psi }\right)d\psi}_{\mathcal{Q}(i|\gamma_A)}\}.
\end{align}
By statistically
averaging the equation above over the PDF of $\gamma_A$, which is given by
\begin{align}
	p_{\gamma_A}(\gamma_A)=\frac{N_0}{(\alpha_A(l))^2\beta_A}\exp\left(\frac{-N_0\gamma_A}{(\alpha_A(l))^2\beta_A}\right),
\end{align}  $\mathcal{Q}(i)$ can be derived as
\begin{align}
	\mathcal{Q}(i)&\!=\!\frac{1}{\pi}\int_{0}^{\frac{\pi}{2}}\!\!\int_{0}^{\infty}\!\!\exp\!\!\left(-\frac{((2i+1)d_A)^2\gamma_A}{\sin^2\psi }\right)\!\! p_{\gamma_A}(\gamma_A)d\gamma_Ad\psi=\frac{1-\mathcal{G}(i)}{2},
\end{align}
where
\begin{align}
	\mathcal{G}(i)=\sqrt{\frac{(\alpha_A(l))^2\beta_A((2i+1)d_A)^2}{N_0+(\alpha_A(l))^2\beta_A((2i+1)d_A)^2}}.
\end{align}
In this manner, we have
\begin{align}
	P_{\rm {bit}}^A(l|m_A)=&\frac{2}{M_A}\sum_{i=0}^{(1-2^{-m_A})M_A-1}\left\{(-1)^{\lfloor \frac{i(2^{m_A}-1)}{M_A}\rfloor}\left(2^{m_A-1}-\lfloor\frac{i(2^{m_A-1})}{M}+\frac{1}{2}\rfloor\right)\mathcal{Q}(i)\right\}.
\end{align}
Therefore, we have
\begin{align}\label{PbitA}
	P_{\rm {bit}}^A(l)=\frac{1}{\log_2(M_A)}\sum_{m_A=1}^{\log_2(M_A)}P_{\rm {bit}}^A(l|m_A).
\end{align}
Finally, by substituting  (\ref{PA_SB}) and (\ref{PbitA}) into (\ref{PA}), we can derive the expression of $P_A$. 
\subsection{Optimal Constellation Analysis}
\begin{figure}[t]
	\centering
	\includegraphics[width=5in]{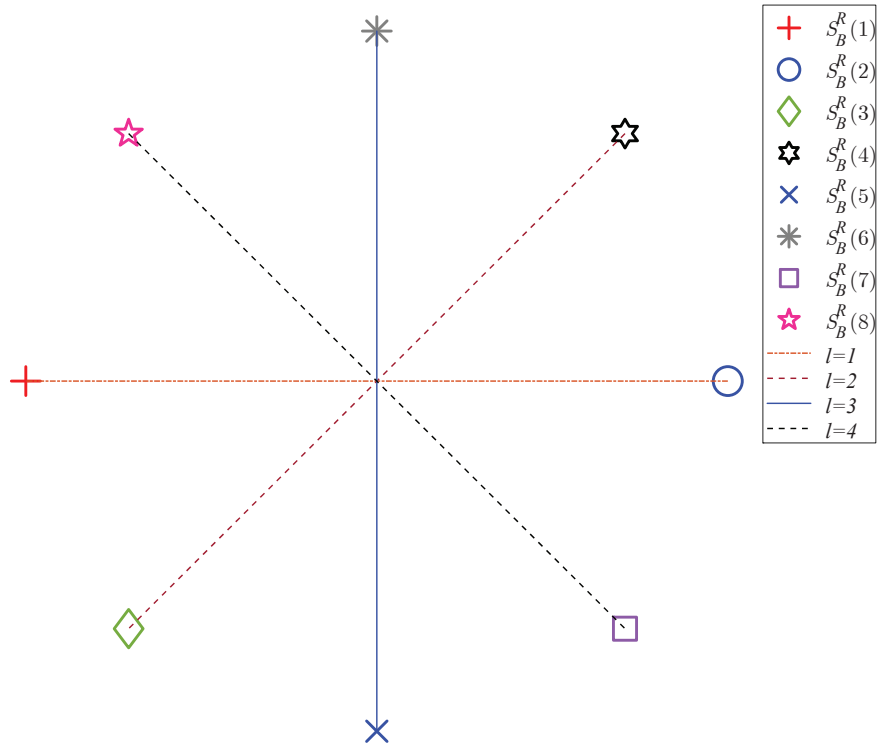}
	\caption{Constellations of $s_B$ with $N=4$, where  $s_B$ is a 2-PAM symbol,  and $\mathcal{S}_B^R$ denotes the union constellation of $s_B$ and PL.}
	\label{constellation_symB_4}
\end{figure} 
In this subsection, we discuss the optimal design of constellation in terms of the values of $N$ as well as $\textbf{G}$. For convenience, let $d_{B,ij}$ denote the Euclidean distance between points $\mathcal{S}_B^R(i)$ and $\mathcal{S}_B^R(j)$, where $i$, $j\in\{1,\ldots,NM_B\}$. According to the geometric analysis, we can obtain
the following.
\begin{lemma}
	Optimal BER performance is achieved at $N = 2$.
\end{lemma}
\begin{proof}
	 Since the cases of $N\geq2$ all have similar derivation processes, we can obtain pattern by analyzing a certain case. For ease of discussion, we consider the  2-PAM signal  $s_A$, 2-PAM signal $s_B$, and $N=4$ scenario as an example. As indicated in \textbf{Algorithm \ref{Detection_algorithm}}, $N$ only affects the order of combined constellation.
	When $N=4$, the combined constellation is shown as Fig. \ref{constellation_symB_4}. Let $d_{B,i,\min}$ denotes the minimum Euclidean distance between $\mathcal{S}_B^R(i)$ and other points. Considering the symmetry of Fig. \ref{constellation_symB_4}, we first take $\mathcal{S}_B^R(1)$ for instance.   From Fig. \ref{constellation_symB_4}, we can obtain $d_{B,1,\min}$ as
	\begin{align}
		d_{B,1,\min}=\min\{d_{B,12},d_{B,13},d_{B,16},d_{B,18}\},
	\end{align}
	where
	\begin{align}
		d_{B,12}=2d_B|\alpha_B(1)|,
	\end{align}
	\begin{align}
		d_{B,13}=\sqrt{|d_B\alpha_B(1)|^2+|d_B\alpha_B(2)|^2-2d_B^2|\alpha_B(1)\alpha_B(2)|\cos\left(\frac{\pi}{N}\right)},
	\end{align}
	\begin{align}
	d_{B,16}=\sqrt{|d_B\alpha_B(1)|^2+|d_B\alpha_B(3)|^2-2d_B^2|\alpha_B(1)\alpha_B(3)|\cos\left(\frac{2\pi}{N}\right)},
   \end{align}
	and
	\begin{align}\label{d14_d23}
		d_{B,18}=\sqrt{|d_B\alpha_B(1)|^2+|d_B\alpha_B(4)|^2-2d_B^2|\alpha_B(1)\alpha_B(4)|\cos\left(\frac{\pi}{N}\right)},
	\end{align}
	respectively. In the similar manner, we can derive $d_{B,2,\min}$, and so on.
	Therefore, the minimum Euclidean distance of Fig. \ref{constellation_symB_4}, can be derived as 
	\begin{align}\label{dBmin}
		d_{B,\min}=\min\{d_{B,1,\min},d_{B,2,\min},d_{B,3,\min},d_{B,4,\min},d_{B,5,\min},d_{B,6,\min},d_{B,7,\min},d_{B,8,\min}\}.
	\end{align}	
	Since we have $N\geq 2$, the angle ${\pi}/{N}$ and ${2\pi}/{N}$ are values between 0 and $\pi$. In other words, $\cos(\pi/N)$ is monotonously increasing with
	respect to $N$. 
	Evidently, to achieve optimal BER performance, $N$ should be equal to the minimum value, which proves \textbf{Lemma 1}.
\end{proof}
\begin{corollary}
Optimal BER performance is achievable when the information of PL is only carried out by the difference of angle. 
\end{corollary}

\begin{proof}
According to \textbf{Lemma 1}, we consider 2-PAM signal  $s_A$, 2-PAM signal $s_B$, and $N=2$ scenario in the following.	
Unlike the derivation of \textbf{Lemma 1}, here we should also consider the BER performance of $U_A$. We begin with the analysis of $s_A$, whose constellation is Fig. \ref{constell_symA}. In this manner, its minimum Euclidean distance can be derived as
\begin{align}\label{dAmin}
	d_{A,\min}=2d_A|\alpha_A(l)|,
\end{align}
where $l\in\mathcal{N}$. Obviously, $d_{A,\min}$ is determined by the minimum value of $|\alpha_A(l)|$. Therefore, if $|\alpha_A(l)|$ with different $l$ has a constant value, $d_{A,\min}$ will be optimal and easy to design. For example, if we consider the case where $|\alpha_A(1)|=0.2$, $|\alpha_A(2)|=0.1$, and the case where $|\alpha_A(1)|=|\alpha_A(2)|=0.2$. It is clear that the latter case has larger $d_{A,\min}$.   Then we shift our focus to $s_B$. The combined constellation is shown as Fig. \ref{constellation_symB}, whose minimum Euclidean distance can be written as
\begin{align}\label{dB_minimum}
	d_{B,\min}=\min\{d_{B,12},d_{B,34},d_{B,14},d_{B,23}\},
\end{align}
where
\begin{align}\label{d_1}
	d_{B,12}=2d_B|\alpha_B(1)|,
\end{align}
\begin{align}\label{d_2}
	d_{B,34}=2d_B|\alpha_B(2)|,
\end{align}
and
\begin{align}\label{dB14}
d_{B,14}=d_{B,23}=\sqrt{|d_B\alpha_B(1)|^2+|d_B\alpha_B(2)|^2}.
\end{align}
Collating (\ref{dB14}) with (\ref{d_1}) and (\ref{d_2}), it is clear that $d_{B,\min}$ is determined by the minimum value of $|\alpha_B(1)|$ and $|\alpha_B(2)|$. In a similar manner, larger $d_{B,\min}$ will be achieved by crafty design in the case of $|\alpha_B(1)|=|\alpha_B(2)|$. Consequently, larger $d_{i,\min}$, $i\in\{A,B\}$, will be achieved when $|\alpha_i(l)|=|\alpha_i(l')|$, where $l,l'\in\mathcal{N}$. In this scenario, the incremental bits of PS-NOMA come from $\bm{\Theta}$, which in turn reduces the design complexity of $\textbf{G}$.
Relevant content is also discussed in Section. IV. 
\end{proof}
\section{Numerical Results}
In this section, we illustrate the BER and achievable rate performance of the proposed scheme through Monte Carlo simulation and numerical results. Since the power coefficients of PS-NOMA are randomly varying, from the fairness perspective, we set the counterpart without PS (denoted by NOMA) as a regular NOMA scheme. The average powers of $h_A$, $h_B$ 
are set to be $\beta_A = 10$, $\beta_B=1$, respectively.
For brevity, we will refer to “PS-NOMA  ($M_A$-PAM, $M_B$-PAM, $N$)” as the PS-NOMA 
scheme where the BS has $N$ PLs and maps
the information bits for $U_A$ and $U_B$ into $M_A$-PAM  and
$M_B$-PAM, respectively. Similarly, “NOMA ($M_A$-PAM, $M_B$-PAM)” denotes the NOMA scheme with $M_A$-PAM and $M_B$-PAM signals intended for $U_A$ and $U_B$, respectively.  Besides, we take SNR=$1/N_0$ as the horizontal axis of the following figures. 
\begin{figure}[t]
	\centering
	\includegraphics[width=5in]{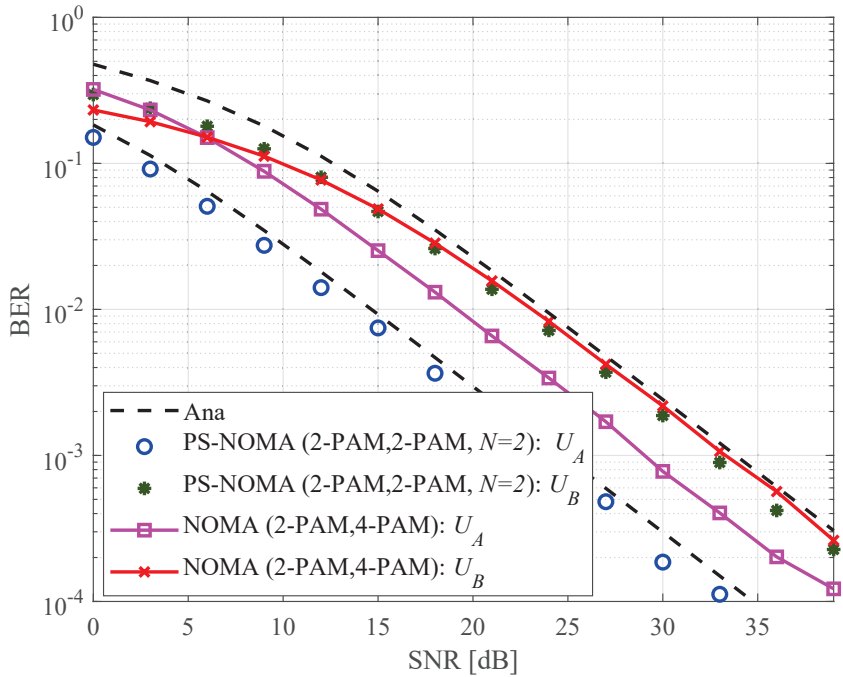}\\
	\caption{BER comparison between PS-NOMA and NOMA.}
	\label{PS_NOMA_VS_RANDOM_NOMA_BER}
\end{figure} 
\begin{figure}
	\centering
	\includegraphics[width=5in]{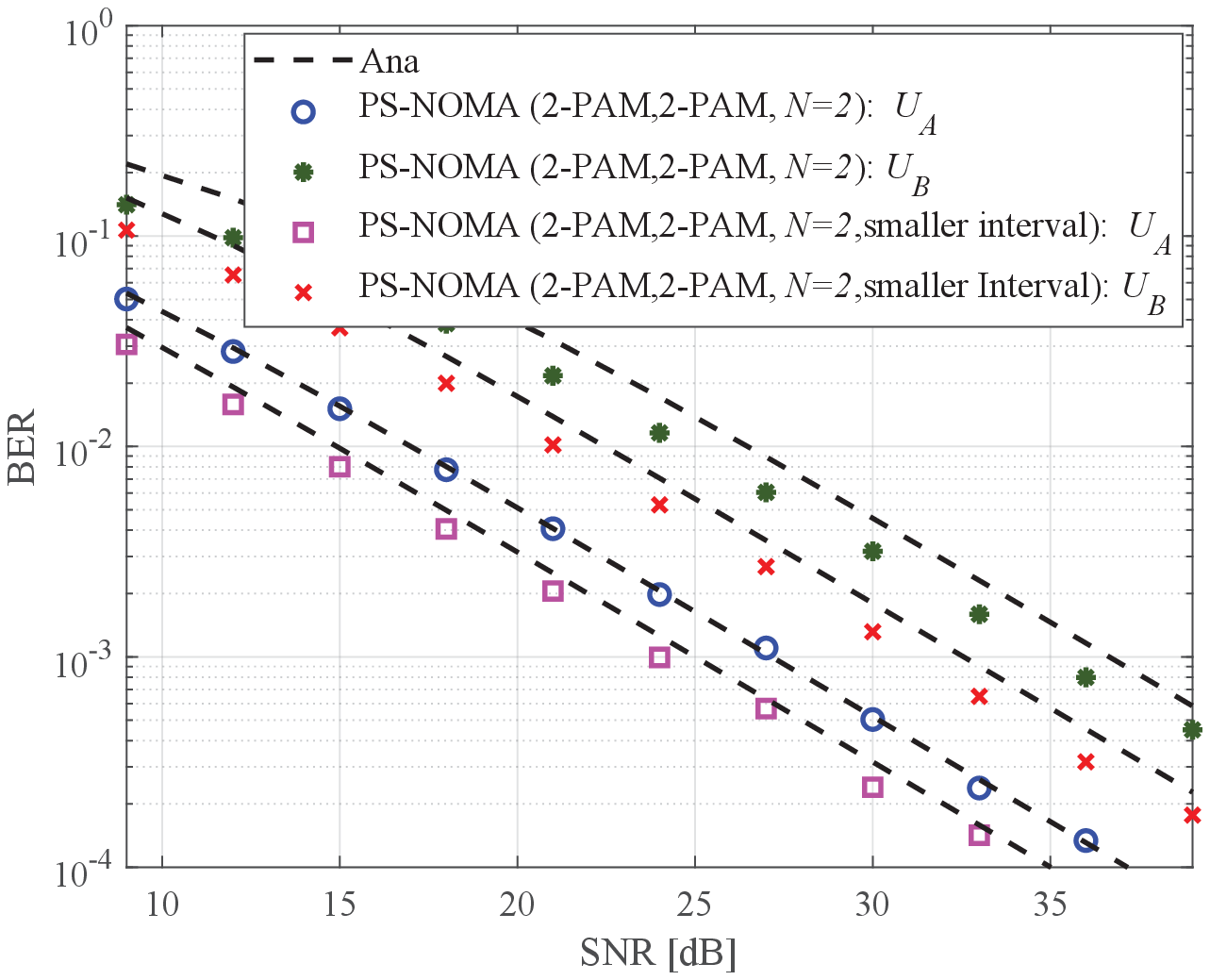}\\
	\caption{BER comparison between PS-NOMA schemes with different PL intervals: Case 1.}
\label{PS_NOMA_VS_LARGER_INTERVAL_uA_same}
\end{figure}
\begin{figure}
	\centering
	\includegraphics[width=5in]{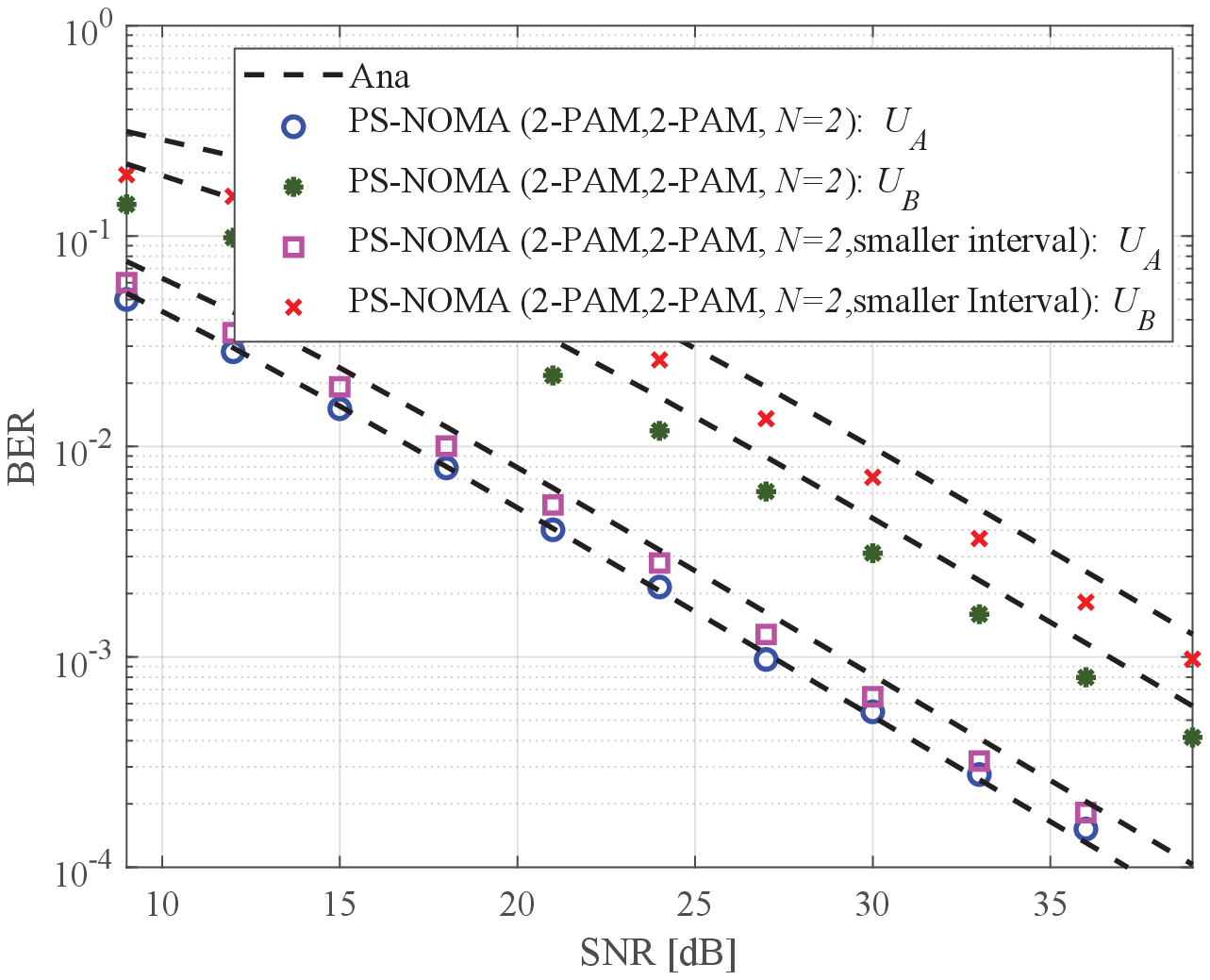}\\
	\caption{BER comparison between PS-NOMA schemes with different PL intervals: Case 2.}
	\label{PS_NOMA_VS_Smaller_INTERVAL_0.05_0.45}
\end{figure}
\begin{figure}
	\centering
	\includegraphics[width=5in]{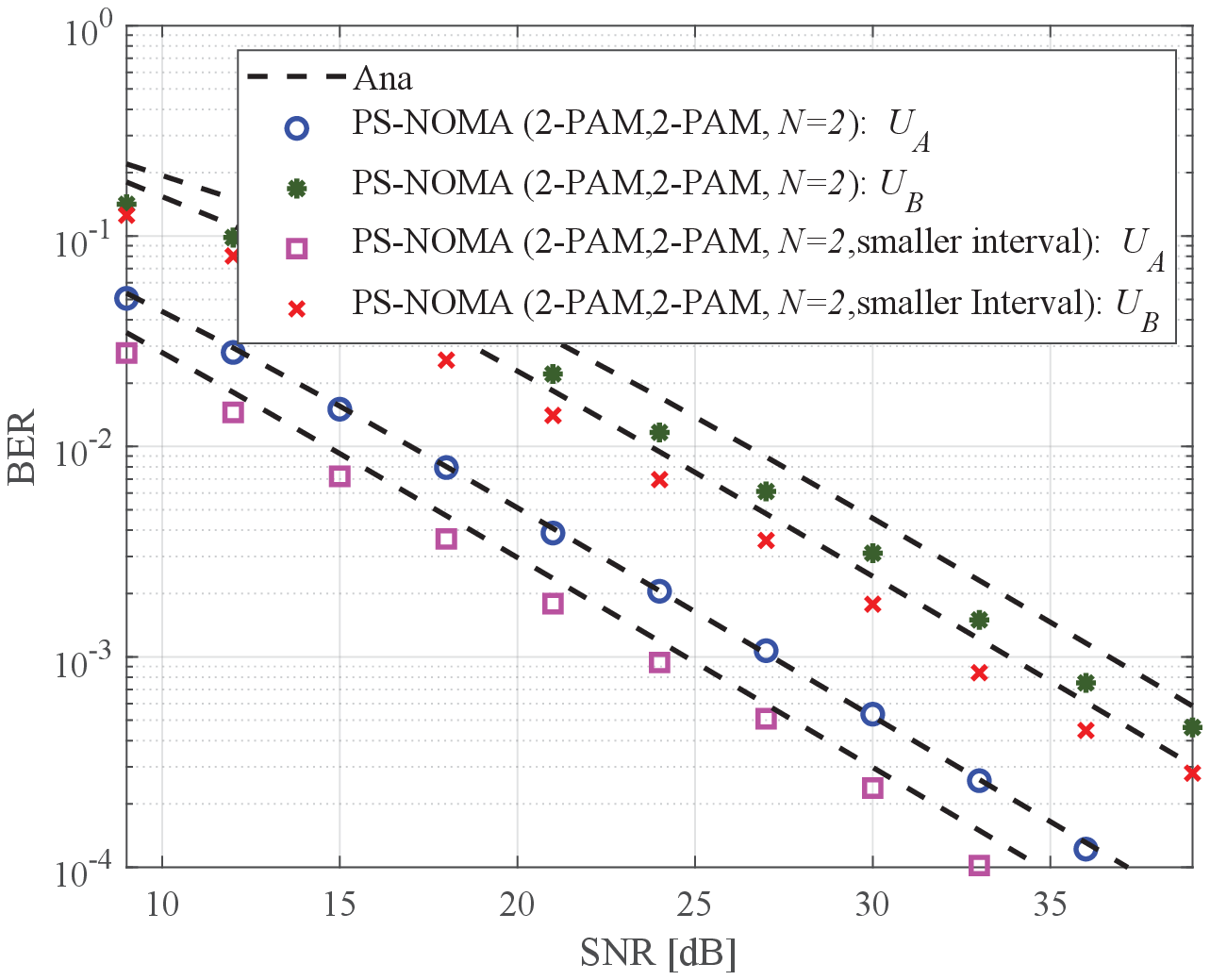}\\
	\caption{BER comparison between PS-NOMA schemes with different PL intervals: Case 3.}
	\label{PS_NOMA_VS_Smaller_INTERVAL_0.2_0.3}
\end{figure}
\begin{figure}
	\centering
	\includegraphics[width=5in]{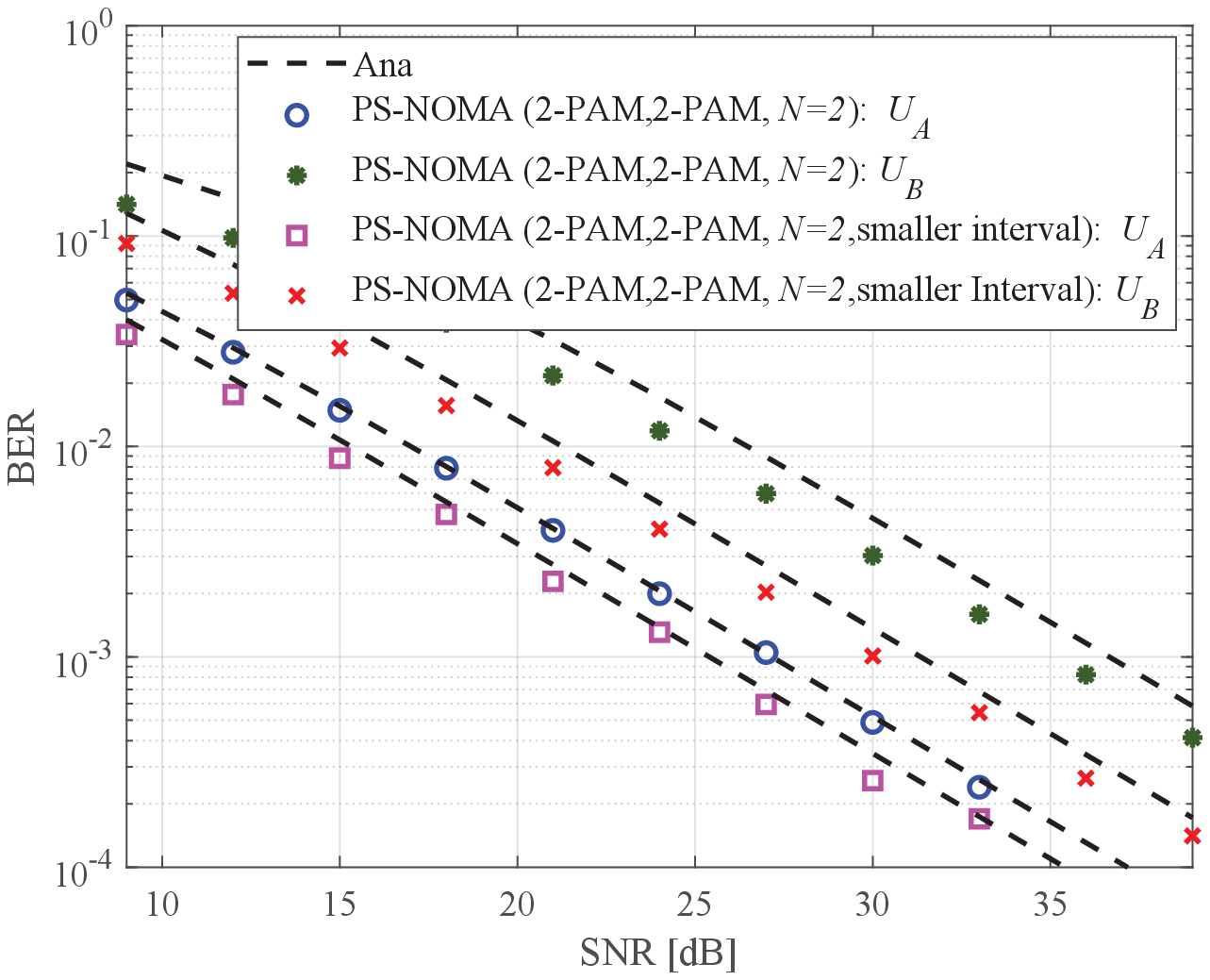}\\
	\caption{BER comparison between PS-NOMA schemes with different PL intervals: Case 4.}
	\label{PS_NOMA_VS_Smaller_INTERVAL_0.2_0.4}
\end{figure}
\subsection{BER Performance}
In this subsection, the BER performance of PS-NOMA and NOMA is compared, assuming that $U_A$ and $U_B$ in all considered schemes employ ML detection. Since the comparison of OMA and NOMA has been widely argued in previous works,   the BER curves of OMA counterparts
are not presented.

In Fig. \ref{PS_NOMA_VS_RANDOM_NOMA_BER}, we compare the BER performance of “PS-NOMA
(2-PAM, 2-PAM, $N=2$)” and “NOMA (2-PAM, 4-PAM)”  with their theoretical
BER results presented. Given fairness, we assume that the maximum spectral efficiency for both schemes are 3 bits per second per Hertz (bps/Hz), and both schemes have the same constellation order. According to $\textbf{Lemma 1}$ and $\textbf{Corollary 1}$, the power matrices of “PS-NOMA
(2-PAM, 2-PAM, $N=2$)” and “NOMA (2-PAM, 4-PAM)” are set by 
\begin{align}\label{sa28}
	\textbf{G}=\left[               
	\begin{array}{cc}   
		\sqrt{0.2},&\sqrt{0.8}\\  
		\sqrt{0.2}\exp(j\pi/2),&\sqrt{0.8}\exp(j\pi/2)\\ 
	\end{array}
	\right]  
\end{align}
and
$\textbf{G}=[\sqrt{0.2},\sqrt{0.8}]$, respectively.
It can be seen from Fig.~\ref{PS_NOMA_VS_RANDOM_NOMA_BER}
that the theoretical curves approximately match with their simulation
counterparts. Besides, one can find that the BER performance of $U_A$ is superior to that
of $U_B$ in both PS-NOMA and NOMA, which comes from the fact $\beta_A>\beta_B$. Although all the curves have the same constellation order as 8, the performances of PS-NOMA are always better than those of NOMA. 
This phenomenon tells us that applying PS can improve performance without increasing constellation complexity when the  amplitudes of power allocation coefficients of the system are the same.  

Figure \ref{PS_NOMA_VS_LARGER_INTERVAL_uA_same}, \ref{PS_NOMA_VS_Smaller_INTERVAL_0.05_0.45}, \ref{PS_NOMA_VS_Smaller_INTERVAL_0.2_0.3}, and \ref{PS_NOMA_VS_Smaller_INTERVAL_0.2_0.4} respectively investigate  four different scenarios where the power matrices of the compared counterparts 
\begin{align}
	\textbf{G}=\left[                
	\begin{array}{cc}   
		\sqrt{0.1},&\sqrt{0.9}\\  
		\sqrt{0.2}\exp(j\pi/2),&\sqrt{0.8}\exp(j\pi/2)\\ 
	\end{array}
	\right]
\end{align} for case 1,  
\begin{align}
	\textbf{G}=\left[                 
	\begin{array}{cc}   
		\sqrt{0.3},&\sqrt{0.7}\\  
		\sqrt{0.4}\exp(j\pi/2),&\sqrt{0.6}\exp(j\pi/2)\\ 
	\end{array}
	\right]
\end{align} for case 2, 
\begin{align}\label{G0.20.2}
\textbf{G}=\left[                 
\begin{array}{cc}   
	\sqrt{0.2},&\sqrt{0.8}\\  
	\sqrt{0.2}\exp(j\pi/2),&\sqrt{0.8}\exp(j\pi/2)\\ 
\end{array}
\right]
\end{align} for case 3, and
\begin{align}\label{G0.20.3}
	\textbf{G}=\left[                 
	\begin{array}{cc}   
		\sqrt{0.1},&\sqrt{0.9}\\  
		\sqrt{0.1}\exp(j\pi/2),&\sqrt{0.9}\exp(j\pi/2)\\ 
	\end{array}
	\right]
\end{align}
for case 4, respectively. Besides, the
power matrix of the benchmark, i.e., PS-NOMA (2-PAM, 2-PAM, $N=2$), equals 
\begin{align}\label{benchmark}
\textbf{G}=\left[                 
\begin{array}{cc}   
	\sqrt{0.1},&\sqrt{0.9}\\  
	\sqrt{0.4}\exp(j\pi/2),&\sqrt{0.6}\exp(j\pi/2)\\ 
\end{array}
\right]. 
\end{align}  Here the PL interval denotes the absolute difference between the absolute values of two adjacent rows of $\textbf{G}$. Smaller interval denotes  smaller value of PL interval than (\ref{benchmark}). Hence all the cases correspond to smaller intervals. 
We first discuss the BER performance of $U_B$.  Concretely, in cases 1, 3, and 4, the minimum Euclidean distances of $s_B$ are larger than those of the benchmark according to (\ref{dB_minimum}). Hence they show better BER performances, while opposite observations are found in the other cases. Furthermore, the phenomenon that performance of case 4 is better than that of benchmark also confirms the correctness of \textbf{Corollary 1}.   More specifically, by comparing the curves of $U_A$ in these cases, we observe that the minor power of $U_A$ does not decrease the BER performance of $U_A$. From  (\ref{ml_UA_sB}) and (\ref{ML_UA_SA}), we know that false detection of $s_B$ and PL may result in the wrong $s_B$, or wrong PL, or wrong $s_B$ and PL. The first detection process (\ref{ml_UA_sB}) is related to both $s_B$ and PL, and it is directly related to the correctness of (\ref{ML_UA_SA}), so the power of $U_B$ has a more significant influence than $U_A$. Notably, we assume that $M_A=M_B$ in this figure. Therefore this conclusion stands. In a word, we can consistently achieve better BER performance through designing power matrices with the same power value when $N=2$ and the design of lower-order matrices are much simpler than higher-order ones. 

\begin{figure}[t]
	\centering
	\includegraphics[width=5in]{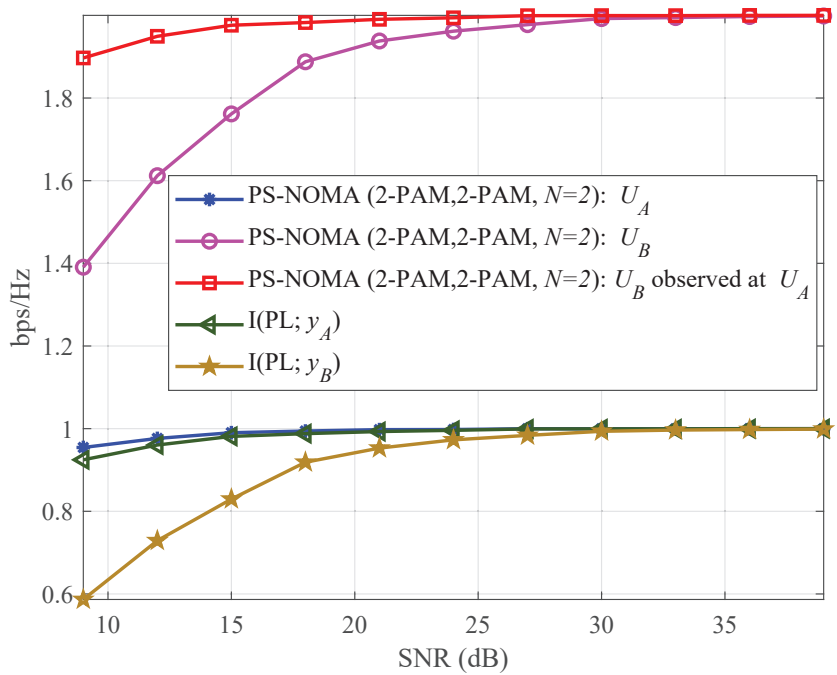}\\
	\caption{Achievable rate and MI performance of PS-NOMA.}
	\label{Capacity_MI_performance}
\end{figure}
\begin{figure}[t]
	\centering
	\includegraphics[width=5in]{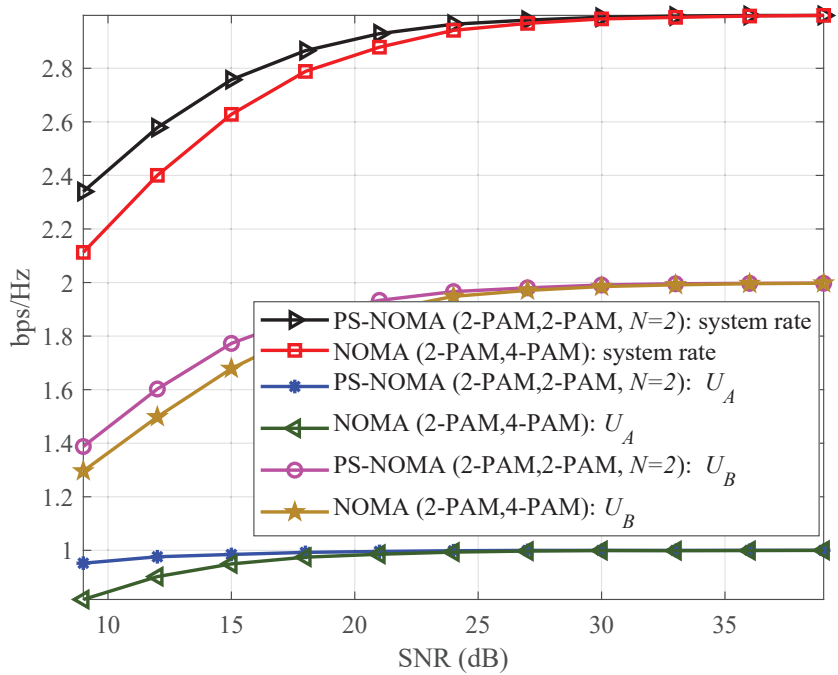}\\
	\caption{Achievable rate comparison between PS-NOMA and NOMA.}
	\label{Capacity_VERUS}
\end{figure}
\subsection{Achievable Rate Performance}
In this subsection, we first evaluate the achievable rate and MI performance of the
proposed PS-NOMA, and then compare
the achievable rate of PS-NOMA with
the NOMA counterparts.

In Fig. \ref{Capacity_MI_performance}, we show the achievable rate and MI performance of  PS-NOMA with (\ref{sa28}) and $\textbf{G}=(\sqrt{0.2},\sqrt{0.8})$, respectively. As seen from the figure, all the curves grow  steadily as the SNR
increases at low-to-medium SNR while achieving floors at high SNR. Obviously, the rates of $U_A$ obtained from (\ref{ISA_YA_PL}), $U_B$ obtained from (\ref{ISBYB}), and $U_B$ observed at $U_A$ obtained from (\ref{ISBYA}) saturate to $\log_2(M_A)$, $\log_2(NM_B)$, and $\log_2(NM_B)$, respectively. On the other hand, the MI curves ${\rm I}({\rm{PL}};y_A)$ and ${\rm I}({\rm{PL}};y_B)$ respectively denote the PL gains at $U_A$ and $U_B$ and are respectively generated from (\ref{IPL_YA}) and (\ref{IPL_YB}). Evidently, both curves saturate  at 1 bps/Hz since the input entropy
of the power-domain is $\log_2(N)$. In other words, all the gains gleaned from the power-domain are assigned to $U_B$. Moreover, since the channel quality
of $U_A$ is much better than that of $U_B$, the achievable rate or MI obtained at $U_A$ is always higher than that at $U_B$, which guarantees the success of SIC. 

Subsequently, in Fig. \ref{Capacity_VERUS}, we present the curves of PS-NOMA versus NOMA with the same system setups as Fig. \ref{Capacity_MI_performance}. Compared to NOMA,
PS-NOMA provides achievable rate gains for both $U_A$ and $U_B$, and obtains the sum achievable rate enhancement in turn. This system performance improvement is more pronounced at low SNR, while the system gap narrows with increasing SNR. Since PS-NOMA and NOMA have the same constellation order, at
high SNR, the achievable rate performances of them are equal. It is worth noting that the performance improvement of $U_A$ is more obvious than that of $U_B$. 
\section{Conclusions}
In this paper,  we have proposed a novel two-user NOMA scheme using PS, in which the ordinary PAM symbol carries the bits of users.  In addition, the BS randomly chooses a PL from the power matrix preset for each transmission to carry more information.  For PS-NOMA,
approximate BER expressions have been derived in
closed form for both users over Rayleigh
flat fading channels. The rates and MI under finite
input have also been developed. Computer simulations have
verified the performance analysis and shown that the proposed
scheme achieves better BER performance than  conventional NOMA without increasing constellation complexity. Moreover, we observe that the proposed
PS-NOMA can achieve better BER performance with the minimum number of PLs, which reduces the design difficulty of the power matrix. Furthermore, if we set an equal  amplitude for each PL, the system performance will also increase. On the other hand, the proposed scheme also shows superior achievable rate performance compared to that  of the NOMA counterpart. 
This work focuses only on the two-user scenario. We leave the more-user
extension for future study.
\begin{appendices}
	\setcounter{equation}{0}
	\renewcommand{\theequation}{\thesection.\arabic{equation}}
	\section{Proof of $(c)$}\label{Proof_of_c} 
To process step $(c)$ of (\ref{PEP_SB_HB}), we use the method introduced in [\ref{pepmethod}]. Let  
$\mathcal{D}\triangleq-|h_B\triangle x_B|^2-2\mathcal{R}\{w_B^*h_B\triangle x_B\}$. With $w_B=h_Bx_A+n_B$, $\mathcal{D}$ can be rewritten as
\begin{align}
\mathcal{D}=&-|h_B\triangle x_B|^2-|h_B|^22\mathcal{R}\{x_A^*\triangle x_B\}-2\mathcal{R}\{n_B^*h_B\triangle x_B\}.
\end{align}
Obviously, $\mathcal{D}$ is Gaussian distributed
with 
\begin{align}
	\mathbb{E}\{\mathcal{D}\}=-|h_B\triangle x_B|^2-|h_B|^22\mathcal{R}\{x_A^*\triangle x_B\}
\end{align}
and
\begin{align}
	\text{Var}\{\mathcal{D}\}=&\text{Var}\{n_B^*h_B\triangle x_B\}=2N_0|h_B\triangle x_B|^2.
\end{align}
Therefore, we have $\Pr(\mathcal{D}>0)=Q({-\mathbb{E}\{\mathcal{D}\}}/{\sqrt{\text{Var}\{\mathcal{D}\}}})$, which can be easily simplified as the result of $(c)$.
\end{appendices}

\end{document}